\documentclass[reprint,pra,aps,longbibliography]{revtex4-2}
\usepackage[english]{babel}
\usepackage{amsmath}
\usepackage{amsthm}
\usepackage{graphicx}
\usepackage{comment}

\def\be{\begin{equation}}
\def\ee{\end{equation}}
\def\ber{\begin{eqnarray}}
\def\eer{\end{eqnarray}}
\def\bern{\begin{eqnarray*}}
\def\eern{\end{eqnarray*}}

\def\rv{\mathbf{r}} 

\def\pv{\mathbf{p}}

\def\Av{\mathbf{A}}
\def\Bv{\mathbf{B}}

\def\Fv{\mathbf{F}}

\def\Rv{\mathbf{R}} 

\def\Vv{\mathbf{V}}

\def\0v{\mathbf{0}}
\def\1v{\mathbf{1}}
\def\2v{\mathbf{2}}
\def\3v{\mathbf{3}}

\def\pa{\partial}

\DeclareMathAlphabet\mathbfcal{OMS}{cmsy}{b}{n}

\def\Rvm{{\underset{\text{\raisebox{3 pt}{=}}}{R}}}
\def\rvm{{\underset{\text{\raisebox{3 pt}{=}}}{r}}}

\newtheorem{theorem}{Theorem}
\newtheorem{lemma}{Lemma}

\begin{document}

\title{In defence of the Ehrenfest mean-field molecular dynamics}

\author{Vladimir~U.~Nazarov}
\affiliation{Fritz Haber Research Center of Molecular Dynamics, the Hebrew University of Jerusalem, Institute of Chemistry,  Israel}
\email{vladimir.nazarov@mail.huji.ac.il}

\begin{abstract}
Re-visiting  the problem of the coupled electronic--nuclear motion, we prove that the so called `mean-field' Ehrenfest picture is an exact theory, rather than an approximation, in the classical limit for nuclei.
By establishing its full equivalence to the much more involved
Exact Factorization approach to the problem (Abedi {\it et al.}, Europhys. Lett. {\bf 106}, 33001 (2014)), our conclusion
rehabilitates the mean-field Ehrenfest molecular dynamics as a simple,  effective, while formally exact scheme, amenable to the use within the framework of the time-dependent density functional theory.
\end{abstract}

\maketitle

Consider a system of $N=N_e+N_n$ particles, $N_e$ of which are electrons and $N_n$ are nuclei.
We are concerned with the force acting on a nucleus within the framework of the mixed quantum (for electrons) -- classical (for nuclei) dynamics. 
The routinely used approach, known as the mean-field Ehrenfest picture, suggests that  nuclei move under the action of the Coulomb force of the instantaneous distribution of the electron density, plus the Coulomb force of the point-charges of other nuclei. The quantum-mechanical motion of the electronic subsystem is governed by the many-body time-dependent (TD) Schr\"{o}dinger equation (TDSE),
where nuclei serve as a source of the  external Coulomb field.
Put in formulae, this scheme reads
\begin{equation}
\begin{split}
\Fv_i(t)
&=   -\nabla_{\Rv_i} \left[ -\int d\rv  \frac{Z_i n_e(\rv,t) }{|\Rv_i-\rv|} \right. \\
& \left. + \sum\limits_{1=j\ne i}^{N_n} \frac{Z_j Z_i}{|\Rv_i-\Rv^{cl}_j(t)|} \right]_{\Rv_i\to \Rv_i^{cl}(t)},
\end{split}
\label{Ehfff}
\end{equation}
where $\Fv_i(t)$ is the force on the $i$-th nucleus, $\Rv_j^{cl}(t)$ are classical trajectories of nuclei,
$Z_j$ are their charges, and $n_e(\rv,t)$ is the  particle-density of electrons (we use atomic units ($m_e=\hbar=|e|=1$) throughout). 
Within this scheme, the electronic TDSE reads
\begin{equation}
\left[i \pa/\pa t -\hat{H}_e(t)\right] \Psi_e(\rvm,t)=0,
\label{EhSE}
\end{equation}
where the  Hamiltonian is
\begin{equation}
\hat{H}_e(t) \! = \! -  \! \sum\limits_{k=1}^{N_e} \frac{1}{2}\nabla_k^2 \! + \! \! \sum\limits_{k>l=1}^{N_e} \frac{1}{|\rv_k-\rv_l|}  
 \! - \! \! \! \sum\limits_{k,j=1}^{N_e,N_n}  \frac{Z_j}{|\rv_k -\Rv_j(t)|},
 \label{EhH}
\end{equation}
$\Psi_e(\rvm,t)$ is the electronic wave-function, and by $\rvm$ and $\Rvm$ we will denote the sets of the electronic and nuclear coordinates, respectively. 
The electron density, needed in Eq.~\eqref{Ehfff} in order to determine the force, is found from the solution of Eq.~\eqref{EhSE} as
\begin{equation}
n_e(\rv,t)= \langle \Psi_e(t)|\hat{n}_e(\rv)|\Psi_e(t)\rangle_\rvm,
\end{equation}
where
\begin{equation}
\hat{n}_e(\rv)=\sum\limits_{k=1}^{N_e} \delta(\rv-\rv_k)
\end{equation}
is the particle-density operator, and $\langle \dots \rangle_\rvm$ stands for the integration over electronic coordinates.

The mean-field Ehrenfest picture is generally considered to be an approximation, not only due to the classical treatment of nuclei, but in the classical limit itself. Advanced methods of the mixed quantum--classical description, designed to improve on the mean-field Ehrenfest scheme, have been developed 
\cite{Tully-90,Tully-98,Ben-Nun-00,Doltsinis-02,Bonella-05,Curchod-13}. 
Among them,  stands out the method based on the Exact Factorization (EF) formalism \cite{Abedi-10,Abedi-12}, which carries out the classical limit procedure without additional assumptions \cite{Abedi-14,Agostini-14}.

In view of the large amount of work having dealt with this problem over years, it comes out as an utter surprise that, within the mixed quantum--classical description, the mean-field Eherenfest dynamics is an exact theory. Nevertheless, this is, indeed, the case, as we show in this paper.

We emphasize that our subject is the strict classical limit for heavy particles. 
Therefore, such effects as trajectories' bifurcation, which retain, at least partly, quantum mechanical features on the side of nuclei, are outside the scope of this study.

{\it I. Ehrenfest force.}--
In this part it is convenient to use the continuous numbering of particles, without explicit separation into electrons and nuclei.
We write down the Hamiltonian of the system
\begin{equation}
\hat{H}= -  \sum\limits_{i=1}^N \frac{1}{2 m_i}\nabla_i^2+ \sum\limits_{i>j=1}^N \frac{z_i z_j}{|\rv_i-\rv_j|}.
\label{H}
\end{equation}
where $m_i$ and $z_i$ are the mass and charge of the $i$-th particle.

\begin{theorem}
Equation \eqref{Ehfff} for the force is exact in the mixed quantum--classical representation.
\label{Tforce}
\end{theorem}
\begin{proof}
Let us select a nucleus, which we will call the 1-st particle. The rate of its  momentum change can be written as \cite{Ehrenfest-27}
\begin{equation}
\begin{split}
\frac{d \pv_1(t)}{d t} &= \frac{d}{d t} \langle \Psi(t)|  \frac{1}{i} \nabla_1|\Psi(t)\rangle\\
&=\langle \Psi(t)|  \left[\hat{H},\nabla_1 \right]|\Psi(t)\rangle,
\end{split}
\end{equation}
where $\Psi(t)$ is the wave-function of the whole composite system. Then, with account of Eq.~\eqref{H}, we can write
\begin{equation}
\begin{split}
\frac{d \pv_1(t)}{d t} \! \!& = \! -\! \sum\limits_{i=2}^N \! \int \! \! d\rv_1 \dots d\rv_N |\Psi(\rv_1,\dots,\rv_N,t)|^2 
\nabla_1 \frac{z_1 z_i}{|\rv_i-\rv_1|} \\
&= \! \sum\limits_{i=2}^N \! \int \! d\rv_1 \dots d\rv_N |\Psi(\rv_1,\dots,\rv_N,t)|^2 
\nabla_i \frac{z_1 z_i}{|\rv_i-\rv_1|} .
\end{split}
\label{F0}
\end{equation}

We now take the classical limit with respect to the 1-st particle, assuming that it is localized in the infinitesimal vicinity to a classical trajectory $\Rv^{cl}_1(t)$.
Then, $|\Psi(\rv_1,\dots,\rv_N,t)|^2$ is zero unless $\rv_1$ coincides with $\Rv^{cl}_1(t)$, which allows us to make the substitution $\rv_1 \to \Rv^{cl}_1(t)$ in the denominator of Eq.~\eqref{F0}
\begin{equation}
\begin{split}
\frac{d \pv_1(t)}{d t} 
 &=   \sum\limits_{i=2}^N \! \int \! d\rv_1 \dots d\rv_N |\Psi(\rv_1,\dots,\rv_N,t)|^2 \\
&\times \nabla_i \frac{z_1 z_i}{|\rv_i-\Rv^{cl}_1(t)|}.
\end{split}
\label{F1}
\end{equation}
The latter equation can be rewritten as
\begin{equation}
\begin{split}
\frac{d \pv_1(t)}{d t} 
& =  \! \sum\limits_{i=2}^N \! \int \! d\rv d\rv_1 \dots d\rv_N z_i \delta(\rv_i-\rv) \\
&\times |\Psi(\rv_1,\dots,\rv_N,t)|^2 
\nabla \frac{z_1}{|\rv-\Rv_1(t)|}.
\end{split}
\label{F2}
\end{equation}
Noting that 
\begin{equation}
\rho^{(N-1)}(\rv,t) \! = \!\sum\limits_{i=2}^N \! \int \! \! d\rv_1 \dots d\rv_N  z_i \delta(\rv_i-\rv) |\Psi(\rv_1,\dots,\rv_N,t)|^2
\end{equation}
is the charge-density of all the particles but the 1-st one,
we have
\begin{equation}
\begin{split}
\frac{d \pv_1(t)}{d t} =   \int d\rv \rho^{(N-1)}(\rv,t)
\nabla \frac{z_1}{|\rv-\Rv_1(t)|}.
\end{split}
\end{equation}

Finally,  taking the classical limit with respect to the remaining $N_n-1$ nuclei, which yields for the charge-density
(we number nuclei first)
\begin{equation}
\rho^{(N-1)}(\rv,t)=-n_e(\rv,t)+\sum\limits_{i=2}^{N_n} z_i \delta[\rv-\Rv^{cl}_i(t)],
\end{equation}
 we arrive at the mean-field Ehrenfest force formula \eqref{Ehfff}.
\end{proof}

A natural question arises: How can Eq.~\eqref{Ehfff} hold true if a much more involved formula for the force,
containing the Berry-connection vector potential and the Berry-curvature magnetic field, 
is known to follow from the EF approach \cite{Abedi-14,Agostini-14}?
In part III we answer this question by re-deriving Eq.~\eqref{Ehfff} directly from the EF result,  demonstrating explicitly  the strict equivalence of the two methods.


{\it II. Electronic equation of motion.}--
The wave-function $\Psi(\Rvm,\rvm,t)$ of the composite electronic--nuclear system  satisfies TDSE
\begin{equation}
\left( i \pa/\pa t -\hat{H} \right) \Psi(\Rvm,\rvm,t)=0,
\label{SE}
\end{equation}
where $\hat{H}$ is the many-body Hamiltonian
\begin{equation}
\hat{H}=  -  \sum\limits_{k=1}^{N_n} \frac{1}{2 M_k}\nabla_{\Rv_k}^2 +\hat{H}^{BO} ,
\label{Hpart}
\end{equation}
and


\begin{equation}
\begin{split}
\hat{H}^{BO} &= -  \! \sum\limits_{k=1}^{N_e} \frac{1}{2}\nabla_{\rv_k}^2 \! + \! \sum\limits_{k>l=1}^{N_e} \frac{1}{|\rv_k-\rv_l|}  \\
&+ \! \sum\limits_{k>l=1}^{N_n} \frac{Z_k Z_l}{|\Rv_k-\Rv_l|} \!
 -\sum\limits_{k,j=1}^{N_e, N_n}  \frac{Z_j}{|\Rv_j-\rv_k|}.
\end{split}
\end{equation}

In the following, we use the machinery of EF formalism \cite{Abedi-10}.  We represent the wave-function as
\begin{equation}
\Psi(\Rvm,\rvm,t)=\chi(\Rvm,t) \Phi_\Rvm(\rvm,t),
\label{EF}
\end{equation}
and impose the partial normalization condition 
\begin{equation}
\langle \Phi_\Rvm|\Phi_\Rvm\rangle_\rvm=1.
\label{norm}
\end{equation}

The variationally best $\chi(\Rvm,t)$, at a given  $\Phi_\Rvm(\rvm,t)$, obeys the equation of motion \cite{Abedi-10}
\begin{equation}
i \dot{\chi}(\Rvm,t) \! = \! \!
\sum\limits_{i=1}^{N_n}\!  \frac{1}{2 M_i} \! \! \left( \! -i \nabla_i \! + \! \Av_i(\Rvm,t) \right)^{\! 2} \! \! \! \chi(\Rvm,t) +\epsilon(\Rvm,t) \chi(\Rvm,t),
\label{chieq}
\end{equation}
where the Berry-connection vector potential is
\begin{equation}
\Av_i(\Rvm,t)=-i \langle \Phi_\Rvm|\nabla_i|\Phi_\Rvm\rangle_\rvm,
\label{Acal}
\end{equation}
and the scalar potential is
\begin{equation}
\begin{split}
\epsilon(\Rvm,t) &=\langle \Phi_\Rvm(\rvm,t)|\hat{H}^{BO}-i \pa/\pa t|\Phi_\Rvm(\rvm,t)\rangle_\rvm \\
&+\sum\limits_{i=1}^{N_n} \left[ \frac{1}{2 M_i} \langle \nabla_i\Phi_\Rvm(\rvm,t)|\nabla_i \Phi_\Rvm(\rvm,t)\rangle_\rvm 
-\frac{\Av^2_i(\Rvm,t)}{2 M_i} \right].
\end{split}
\label{epseq}
\end{equation}

\newpage

We will need
\begin{lemma}
For the conditional electronic  wave-function  $\Phi_\Rvm(\rvm,t)$ and the corresponding nuclear wave-function $\chi(\Rvm,t)$,
the equality holds
\label{L}
\begin{equation}
\begin{split}
&|\chi(\Rvm,t)|^2 \! \left\{ \left[  i \pa/\pa t \! - \! \hat{H}^{BO} \! - \! \langle \Phi_\Rvm| i \pa/\pa t \! - \! \hat{H}^{BO}|\Phi_\Rvm\rangle_\rvm \right] \! \! \Phi_\Rvm(\rvm,t) \right.
 \\
&\left. + \sum\limits_{i=1}^{N_n} \frac{1}{2 M_i} [ \nabla^2_i -\langle \Phi_\Rvm|\nabla^2_i|\Phi_\Rvm\rangle_\rvm] \Phi_\Rvm(\rvm,t)\right\} \\
&+ \sum\limits_{i=1}^{N_n} \frac{i}{ M_i} \chi^*(\Rvm,t) \nabla_i \chi(\Rvm,t) \cdot [-i\nabla_i- \Av_i(\Rvm,t)] \Phi_\Rvm(\rvm,t) \! = \! 0. 
\end{split}
\label{maineq}
\end{equation}
\end{lemma}

Proof of Lemma \ref{L} is given in Appendix \ref{PL}. We note that, while  Eq.~\eqref{maineq} is fully equivalent to the equation of motion for $\Phi_\Rvm(\rvm,t)$ \cite{Abedi-10}, in its present form it is particularly convenient for our purposes.

\

We continue by considering the integral, which, by Lemma \ref{L}, is equal to zero

\begin{equation}
\begin{split}
& \int d \Rvm |\chi(\Rvm,t)|^2 \\
&\times \left\{  \left[  i \pa/\pa t -\hat{H}^{BO} -\langle \Phi_\Rvm| i \pa/\pa t -\hat{H}^{BO}|\Phi_\Rvm\rangle_\rvm \right] \Phi_\Rvm(\rvm,t) \right.
 \\
& + \sum\limits_{i=1}^{N_n} \frac{1}{2 M_i} [ \nabla^2_i -\langle \Phi_\Rvm|\nabla^2_i|\Phi_\Rvm\rangle_\rvm] \Phi_\Rvm(\rvm,t) \\
&\left. +  \sum\limits_{i=1}^{N_n} \frac{1}{ M_i}   \nabla_i S(\Rvm,t) \cdot [i\nabla_i+ \Av_i(\Rvm,t)] \Phi_\Rvm(\rvm,t)\right\}\\
&- \int d \Rvm  \sum\limits_{i=1}^{N_n} \frac{i}{2  M_i}   \nabla_i |\chi(\Rvm,t)|^2  \cdot [i\nabla_i+ \Av_i(\Rvm,t)] \Phi_\Rvm(\rvm,t)=0,
\end{split}
\label{olast}
\end{equation}

where we use the representation
\begin{equation}
\chi(\Rvm,t)=|\chi(\Rvm,t)| e^{i S(\Rvm,t)}.
\label{chiS}
\end{equation}
After  the integration by parts in the last term of Eq.~\eqref{olast} and simplifications, we can write
\begin{widetext}
\begin{equation}
\begin{split}
& \int d \Rvm |\chi(\Rvm,t)|^2 
\left\{  \left[  i \pa/\pa t -\hat{H}^{BO} -\langle \Phi_\Rvm| i \pa/\pa t -\hat{H}^{BO}|\Phi_\Rvm\rangle_\rvm \right] \Phi_\Rvm(\rvm,t) \right.
 \\
&\left. + \sum\limits_{i=1}^{N_n} \frac{1}{2 M_i}  \langle \nabla_i\Phi_\Rvm|\nabla_i\Phi_\Rvm\rangle_\rvm \Phi_\Rvm(\rvm,t) 
+ \sum\limits_{i=1}^{N_n} \frac{1}{ M_i}   \nabla_i S(\Rvm,t) \cdot [i\nabla_i+ \Av_i(\Rvm,t)] \Phi_\Rvm(\rvm,t)\right\}=0.
\end{split}
\label{MMM}
\end{equation}
\end{widetext}
Equation \eqref{MMM} is the general quantum mechanical result.
In the classical limit for nuclei, we  write \cite{Abedi-14,Agostini-14}
\begin{align}
&|\chi^{cl}(\Rvm,t)|^2=\delta[\Rvm-\Rvm^{cl}(t)],  \label{cl1}\\
&M_i \dot{\Rv}^{cl}_i(t)=\left.\nabla_i S^{cl}(\Rvm,t) +\Av_i^{cl}(\Rvm,t)\right|_{\Rvm\to \Rvm^{cl}(t)}, \label{cl2}
\end{align}
where $\Rvm^{cl}(t)$ are classical trajectories of nuclei, and all the quantities with the `$cl$` superscript are those concerted with the limit \eqref{cl1}-\eqref{cl2}.
In particular, $\Phi_\Rvm^{cl}(\rvm,t)$ denotes the conditional electronic wave-function after the classical limit for nuclei is taken, but yet before the substitution $\Rvm\to\Rvm^{cl}(t)$
is made. 
We will need the following
\begin{lemma}
\label{LL}
\begin{equation}
\left. i\nabla_i \Phi_\Rvm^{cl}(\rvm,t) + \Av^{cl}_i(\Rvm,t) \Phi_\Rvm^{cl}(\rvm,t)\right|_{\Rvm\to \Rvm^{cl}(t)} = \0v.
\end{equation}
\end{lemma}
Proof of the Lemma is given in Appendix \ref{PLL}.

\

After the substitution of Eq.~\eqref{cl1} in Eq.~\eqref{MMM} and using  Lemma \ref{LL}, we have
\begin{equation}
\begin{split}
&\left\{ \left[  i \pa/\pa t -\hat{H}^{BO} -\langle \Phi^{cl}_\Rvm| i \pa/\pa t -\hat{H}^{BO}|\Phi^{cl}_\Rvm\rangle_\rvm  \right. \right. \\
&\left. \left.  +   \sum\limits_{i=1}^{N_n} \frac{1}{2 M_i}  \langle \nabla_i\Phi^{cl}_\Rvm|\nabla_i \Phi^{cl}_\Rvm\rangle_\rvm \right]\Phi^{cl}_\Rvm(\rvm,t) 
\right\}_{\Rvm\to\Rvm^{cl}(t)}=0.
\label{TTT}
\end{split}
\end{equation}
We note that, by the rule for the full derivative,
\begin{equation}
\begin{split}
&\left[ \pa/\pa t  \Phi^{cl}_\Rvm(\rvm,t) \right]_{\Rvm\to\Rvm^{cl}(t)} = \pa/\pa t \Phi^{cl}_{\Rvm^{cl}(t)}(\rvm,t) \\
&- \sum\limits_{i=1}^{N_n}  \dot{\Rv}^{cl}_i(t)\cdot \left\{\nabla_i \Phi^{cl}_\Rvm(\rvm,t)\right\}_{\Rvm\to\Rvm^{cl}(t)},
\end{split}
\label{W4}
\end{equation}
which, with the use of Lemma~\ref{LL}, can be written as
\begin{equation}
\begin{split}
&\left[ i\pa/\pa t  \Phi^{cl}_\Rvm(\rvm,t) \right]_{\Rvm\to\Rvm^{cl}(t)} = i \pa/\pa t \Phi^{cl}_{\Rvm^{cl}(t)}(\rvm,t) \\
&+ \sum\limits_{i=1}^{N_n}  \dot{\Rv}^{cl}_i(t)\cdot  \Av_i^{cl}(\Rvm^{cl}(t),t)\Phi^{cl}_{\Rvm^{cl}(t)}(\rvm,t),
\end{split}
\end{equation}
where $\left. \Phi^{cl}_{\Rvm^{cl}(t)}(\rvm,t)=\Phi^{cl}_\Rvm(\rvm,t)\right|_{\Rvm\to\Rvm^{cl}(t)}$.
Finally, Eq.~\eqref{TTT} takes the form
\begin{equation}
\begin{split}
& \left[  i \pa/\pa t -\hat{H}^{BO}(t)  +C(t) \right] \Phi^{cl}_{\Rvm^{cl}(t)}(\rvm,t) =0,
\label{TTT2}
\end{split}
\end{equation}
where
\begin{equation}
\hat{H}^{BO}(t)= -  \! \sum\limits_{k=1}^{N_e} \frac{1}{2}\nabla_{\rv_k}^2 \! + \! \sum\limits_{k>l=1}^{N_e} \frac{1}{|\rv_k-\rv_l|} 
 -\sum\limits_{k=1}^{N_e} \sum\limits_{l=1}^{N_n} \frac{Z_l}{|\Rv_l(t)-\rv_k|},
\end{equation}
and 
\begin{equation}
\begin{split}
C(t)&=-\langle \Phi^{cl}_{\Rvm^{cl}(t)}| i \pa/\pa t -\hat{H}^{BO}|\Phi^{cl}_{\Rvm^{cl}(t)}\rangle_\rvm \\
 & +   \sum\limits_{i=1}^{N_n} \frac{1}{2 M_i}  \Av^{cl}_i(\Rvm^{cl}(t),t)^2  
\end{split}
\end{equation}
is a real time-dependent constant which, affecting the phase of $\Phi_{\Rvm^{cl}(t)}(\rvm,t)$ only, does not change the electron density $n_e(\rv,t)$, the latter necessary for the evaluation of forces by Eq.~\eqref{Ehfff}.

Importantly, since  in Eq.~\eqref{Ehfff} electrons are represented by their density $n_e(\rv,t)$ only, and  since in Eq.~\eqref{EhH} the  moving nuclei are present by the external classical Coulomb field they create, the direct use of the time-dependent density-functional theory (TDDFT) \cite{Runge-84,Gross-85} is fully justified.
While TDDFT is routinely used within the mean-field Ehrenfest dynamics (e.g, Refs.~\cite{Pruneda-07,Nazarov-14,Nazarov-26} are typical examples),  our finding shows that this does not involve additional approximations in the mixed quantum--classical description.

We, further, emphasize that our findings refer to mixed quantum-classical dynamics only, and they do not affect fully quantum mechanical approaches to the problem of the coupled electronic-nuclear motion \cite{Li-22}.

{\it III. Direct proof of the equivalence between EF and the mean-field Ehrenfest forces.}-- Although by Theorem \ref{Tforce} we have proven the exact validity of Eq.~\eqref{Ehfff}, it is instructive to explicitly show its strict equivalence to the corresponding EF result. The latter reads \cite{Abedi-14,Agostini-14}

\begin{equation}
\begin{split}
&\Fv_i(t)=-\nabla_i \epsilon^{cl}(\Rvm,t)+\frac{\pa \Av^{cl}_i(\Rvm,t)}{\pa t} \\
&\left. -\Vv^{cl}_i(t) \times \Bv_{i i}(\Rvm,t) + \sum\limits_{i'\ne i} \mathbfcal{F}_{i i'}(\Rvm,t) \right|_{\Rvm\to \Rvm^{cl}(t)}, \\
& \
\end{split}
\label{Fcal}
\end{equation}
where $\Vv^{cl}_i(t)= d \Rv^{cl}_i(t)/d t$ is the velocity of the $i$-th nucleus, 
\begin{equation}
\Bv_{i i'}(\Rvm,t)= \nabla_i \times \Av^{cl}_{i'}(\Rvm,t),
\end{equation}
and
\begin{equation}
\begin{split}
&\mathbfcal{F}_{i i'}(\Rvm,t)=-\Vv^{cl}_{i'} \times \Bv_{i i'}(\Rvm,t)\\
&+\left[ (\Vv^{cl}_{i'} \cdot \nabla_{i'}) \Av^{cl}_i(\Rvm,t)-(\Vv^{cl}_{i'} \cdot \nabla_{i}) \Av^{cl}_{i'}(\Rvm,t)  \right].
\end{split}
\end{equation}
We prove
\begin{lemma}
Let $X_{i\alpha}$ be the $\alpha$-th Cartesian coordinate of the $i$-th nucleus. Then
\begin{equation}
\left.\frac{\pa A^{cl}_{i\alpha}(\Rvm,t)}{\pa X_{j\beta}}-\frac{\pa A^{cl}_{j\beta}(\Rvm,t)}{\pa X_{i\alpha}}\right|_{\Rvm\to \Rvm^{cl}(t)} =0,
\label{36}
\end{equation}
\label{LLL}
\end{lemma}

Proof of the Lemma is given in Appendix \ref{PLLL}.

\

For comparative simplicity, below we adhere to the  field-free case, although arbitrary electromagnetic fields can be included along the lines of Ref.~\cite{Nazarov-26-2}. 
Due to Lemma \ref{LLL}, the two first terms on RHS of Eq.~\eqref{Fcal} contribute only. 
For them, after lengthy but necessary algebra (see Appendix \ref{App}), we find 
\begin{widetext}
\begin{equation}
\begin{split}
  \left.-\nabla_i \epsilon^{cl}(\Rvm,t)+\frac{\pa \Av^{cl}_i(\Rvm,t)}{\pa t}\right|_{\Rvm\to \Rvm^{cl}(t)}  
&= \left.  - \langle \Phi^{cl}_\Rvm|[\nabla_i\hat{H}^{BO}]|  \Phi^{cl}_\Rvm\rangle_\rvm \right|_{\Rvm\to \Rvm^{cl}(t)} \\
&=  \left.\nabla_i \int d\rv  \frac{Z_i n_e^{cl}(\rv,t) }{|\Rv_i-\rv|} -\nabla_i \sum\limits_{1=j\ne i}^{N_n} \frac{Z_j Z_i}{|\Rv_j-\Rv_i|} \right|_{\Rvm\to \Rvm^{cl}(t)},
\end{split}
\label{Llast}
\end{equation}
\end{widetext}
which coincides with the Ehrenfest force of Eq.~\eqref{Ehfff}.

\

In conclusion, despite the generally accepted view that the mean-field Ehrenfest dynamics presents an approximate  solution to the problem of the mixed quantum-classical motion of interacting electrons and nuclei, we have proven that this solution is exact. Apart from dispelling a consequential misconception in theoretical many-body physics, 
our finding changes the status of
the Ehrenfest scheme to an efficient practical method in molecular dynamics, which is not compromised by the use of additional approximations.
In particular, the routinely practised  time-dependent density functional theory approach within the mean-field Ehrenfest dynamics 
is shown to be a conceptually exact  rather than approximate method.

\acknowledgements
We thank T. N. Todorov for valuable discussions

%

\newpage

\appendix

\begin{widetext}
\section{Proof of Lemma \ref{L}}
\label{PL}

We can write
\begin{equation}
\begin{split}
0 \! = \! \left(   i \pa/\pa t \! - \! \hat{H} \right)\chi(\Rvm,t) \Phi_\Rvm(\rvm,t) = i \dot{\chi}(\Rvm,t) \Phi_\Rvm(\rvm,t) \! + \! i \chi(\Rvm,t) \dot{\Phi}_\Rvm(\rvm,t) 
\! - \! \chi(\Rvm,t) \hat{H}^{BO} \Phi_\Rvm(\rvm,t)
\! + \! \sum\limits_{i=1}^{N_n} \frac{1}{2 M_i} \nabla^2_i [ \chi(\Rvm,t) \Phi_\Rvm(\rvm,t)],
\end{split}
\end{equation}
and, using Eqs.~\eqref{chieq}-\eqref{epseq}, 
\begin{equation}
\begin{split}
&
 -i \Phi_\Rvm(\rvm,t) \sum\limits_{i=1}^{N_n} \frac{1}{ 2 M_i} [\Av_i(\Rvm,t)\cdot \nabla_i +\nabla_i\cdot \Av_i(\Rvm,t)] \chi(\Rvm,t)   \\
& + \chi(\Rvm,t) \Phi_\Rvm(\rvm,t) \langle \Phi_\Rvm(\rvm,t)|\hat{H}^{BO}-i \pa/\pa t|\Phi_\Rvm(\rvm,t)\rangle_\rvm 
+ \chi(\Rvm,t) \Phi_\Rvm(\rvm,t) \sum\limits_{i=1}^{N_n}  \frac{1}{2 M_i} \langle \nabla_i\Phi_\Rvm(\rvm,t)|\nabla_i \Phi_\Rvm(\rvm,t)\rangle_\rvm
 \\
&+i \chi(\Rvm,t) \dot{\Phi}_\Rvm(\rvm,t)- \chi(\Rvm,t) \hat{H}^{BO} \Phi_\Rvm(\rvm,t) + \chi(\Rvm,t) \sum\limits_{i=1}^{N_n} \frac{1}{2 M_i} \nabla^2_i  \Phi_\Rvm(\rvm,t)
+ \sum\limits_{i=1}^{N_n} \frac{1}{ M_i} \nabla_i \chi(\Rvm,t) \cdot \nabla_i \Phi_\Rvm(\rvm,t)=0.
\end{split}
\end{equation}
\end{widetext}
The latter equation can be rewritten as Eq.~\eqref{maineq}.

\section{Proof of Lemma \ref{LL}}
\label{PLL}

In the classical limit, we multiply Eq.~\eqref{maineq} by $[\Rv_j-\Rv_j^{cl}(t)]$ and integrate over $\Rvm$. Then the only  non-zero term remaining is
\begin{equation}
\begin{split}
 \sum\limits_{i=1}^{N_n} \frac{i}{2  M_i} \int d\Rvm  [\Rv_j-\Rv_j^{cl}(t)] \nabla_i |\chi^{cl}(\Rvm,t)|^2 \\
 \cdot [-i\nabla_i- \Av^{cl}_i(\Rvm,t)] \Phi^{cl}_\Rvm(\rvm,t)=\0v,
\end{split}
\label{maineq2}
\end{equation}
or, after the integration by parts, 
\begin{equation}
\begin{split}
 \int d\Rvm   |\chi^{cl}(\Rvm,t)|^2  [-i\nabla_j- \Av_j^{cl}(\Rvm,t)] \Phi_\Rvm^{cl}(\rvm,t)=\0v.
\end{split}
\label{maineq3}
\end{equation}
With account of Eq.~\eqref{cl1}, Eq.~\eqref{maineq3}  proves the Lemma.

\section{Proof of Lemma \ref{LLL}}
\label{PLLL}

Consider the quantity $\left\langle \frac{\pa \Phi_\Rvm}{\pa X_{j\beta}}\right| \left. \frac{\pa \Phi_\Rvm}{\pa X_{i\alpha}}\right\rangle_\rvm$. On the one hand, we can write
\begin{equation}
\begin{split}
&\left\langle \frac{\pa \Phi_\Rvm}{\pa X_{j\beta}}\right| \left. \frac{\pa \Phi_\Rvm}{\pa X_{i\alpha}}\right\rangle_\rvm   = 
\frac{\pa }{\pa X_{j\beta}}\left \langle \Phi_\Rvm \right| \left.\frac{\pa \Phi_\Rvm}{\pa X_{i\alpha}}\right\rangle_\rvm \\
&-\left \langle \Phi_\Rvm \right| \left.\frac{\pa^2 \Phi_\Rvm}{\pa X_{i\alpha} \pa X_{j\beta}}\right\rangle_\rvm =
i \frac{\pa A_{i\alpha}(\Rvm,t)}{\pa X_{j\beta}} 
-\left \langle \Phi_\Rvm \right| \left.\frac{\pa^2 \Phi_\Rvm}{\pa X_{i\alpha} \pa X_{j\beta}}\right\rangle_{\! \! \!\rvm} 
\end{split}
\end{equation}
On the other
\begin{equation}
\begin{split}
&\left\langle \frac{\pa \Phi_\Rvm}{\pa X_{j\beta}}\right| \left. \frac{\pa \Phi_\Rvm}{\pa X_{i\alpha}}\right\rangle_\rvm = 
\frac{\pa }{\pa X_{i\alpha}}\left\langle \frac{\pa \Phi_\Rvm}{\pa X_{j\beta}} \right| \left. \Phi_\Rvm \right\rangle_\rvm \\
&-\left \langle \frac{\pa^2 \Phi_\Rvm}{\pa X_{i\alpha} \pa X_{j\beta}} \right| \left.  \Phi_\Rvm\right\rangle_\rvm \! = \!
-i \frac{\pa A_{j\beta}(\Rvm,t)}{\pa X_{i\alpha}} 
\! - \! \left \langle \frac{\pa^2 \Phi_\Rvm}{\pa X_{i\alpha} \pa X_{j\beta}} \right| \left.  \Phi_\Rvm\right\rangle_{\! \! \!\rvm }
\end{split}
\end{equation}
Adding these two equations and dividing by two, we have
\begin{equation}
\begin{split}
\left\langle \frac{\pa \Phi_\Rvm}{\pa X_{j\beta}}\right| \left. \frac{\pa \Phi_\Rvm}{\pa X_{i\alpha}}\right\rangle_\rvm &=
\frac{i}{2} \left[ \frac{\pa A_{i\alpha}(\Rvm,t)}{\pa X_{j\beta}}-\frac{\pa A_{j\beta}(\Rvm,t)}{\pa X_{i\alpha}}  \right] \\
&-\Re \left \langle \Phi_\Rvm \right| \left.\frac{\pa^2 \Phi_\Rvm}{\pa X_{i\alpha} \pa X_{j\beta}}\right\rangle_{\! \! \!\rvm}. \\
& \ \\
& \
\end{split}
\end{equation}

Hence, using Lemma~\ref{LL}, we can write
\begin{equation}
\begin{split}
&\left. A^{cl}_{i\alpha}(\Rvm,t) A^{cl}_{j\beta}(\Rvm,t)\right|_{\Rvm\to \Rvm^{cl}(t)} =\frac{i}{2} \times\\
&
\left[ \frac{\pa A^{cl}_{i\alpha}(\Rvm,t)}{\pa X_{j\beta}} \! - \! \frac{\pa A^{cl}_{j\beta}(\Rvm,t)}{\pa X_{i\alpha}}  \right]
\! \! - \! \! \left. \Re \left \langle \! \Phi^{cl}_\Rvm \right| \left.\frac{\pa^2 \Phi^{cl}_\Rvm}{\pa X_{i\alpha} \pa X_{j\beta}}\right\rangle_{\! \! \! \rvm }\right|_{\Rvm\to \Rvm^{cl}(t)}
\label{40}
\end{split}
\end{equation}
In Eq.~\eqref{40}, LHS and the second term on RHS are real, while the first term on RHS is imaginary, which makes the validity of Eq.~\eqref{36} evident.

\section{Proof of Eq.~\eqref{Llast}}

\label{App}

For the two remaining non-zero terms in Eq.~ \eqref{Fcal} we write,
according to Eqs.~\eqref{Acal} and \eqref{epseq},
\begin{equation}
\begin{split}
\frac{\pa \Av_i(\Rvm,t)}{\pa t} &=
-i \langle \Phi_\Rvm| \nabla_i|\dot{\Phi}_\Rvm\rangle_\rvm -i \langle \dot{\Phi}_\Rvm|\nabla_i|\Phi_\Rvm\rangle_\rvm \\
&=2 \Im \langle \Phi_\Rvm|\nabla_i |\dot{\Phi}_\Rvm\rangle_\rvm +i \nabla_i \langle \Phi_\Rvm| \dot{\Phi}_\Rvm\rangle_\rvm, 
\end{split}
\end{equation}

\begin{equation}
\begin{split}
&\nabla_i \epsilon(\Rvm,t)=\nabla_i \langle \Phi_\Rvm|\hat{H}^{BO}|\Phi_\Rvm\rangle_\rvm- i \nabla_i \langle \Phi_\Rvm|\dot{\Phi}_\Rvm\rangle_\rvm \\
&+\nabla_i \sum\limits_{j=1}^{N_n} \left[ \frac{1}{2 M_j} \langle \nabla_j\Phi_\Rvm|\nabla_j \Phi_\Rvm\rangle_\rvm 
-\frac{\Av^2_j(\Rvm,t)}{2 M_j} \right].
\end{split}
\end{equation}
Then

\vspace{0.5 cm}

\begin{widetext}

\begin{equation}
\begin{split}
&-\nabla_i \epsilon(\Rvm,t)+\frac{\pa \Av_i(\Rvm,t)}{\pa t} =- \langle \Phi_\Rvm|\hat{H}^{BO}|\nabla_i  \Phi_\Rvm\rangle_\rvm - \langle\nabla_i \Phi_\Rvm|\hat{H}^{BO}| \Phi_\Rvm\rangle_\rvm - \langle \Phi_\Rvm|[\nabla_i\hat{H}^{BO}]|  \Phi_\Rvm\rangle_\rvm\\
&-\nabla_i \sum\limits_{j=1}^{N_n} \left[ \frac{1}{2 M_j} \langle \nabla_j\Phi_\Rvm|\nabla_j \Phi_\Rvm\rangle_\rvm 
-\frac{\Av^2_j(\Rvm,t)}{2 M_j} \right] 
 +2 \Im  \langle \Phi_\Rvm|  \nabla_i \dot{\Phi}_\Rvm\rangle_\rvm -2 \Im \nabla_i \langle \Phi_\Rvm| \dot{\Phi}_\Rvm\rangle_\rvm .
\end{split}
\end{equation}

According to Eq.~\eqref{maineq}
\begin{equation}
\begin{split}
&
    \dot{\Phi}_\Rvm(\rvm,t)= -i\hat{H}^{BO}\Phi_\Rvm(\rvm,t) 
  +i \langle \Phi_\Rvm|  \hat{H}^{BO}|\Phi_\Rvm\rangle_\rvm \Phi_\Rvm(\rvm,t) 
  +\langle \Phi_\Rvm |\dot{\Phi}_\Rvm\rangle_\rvm \Phi_\Rvm(\rvm,t) 
 \\
&+  \sum\limits_{i=1}^{N_n} \frac{i}{2 M_i} [ \nabla^2_i\Phi_\Rvm(\rvm,t) -\langle \Phi_\Rvm|\nabla^2_i|\Phi_\Rvm\rangle_\rvm \Phi_\Rvm(\rvm,t)]  
+ \sum\limits_{i=1}^{N_n} \frac{1}{ M_i} \frac{\nabla_i \chi(\Rvm,t)}{\chi(\Rvm,t)} \cdot [i\nabla_i \Phi_\Rvm(\rvm,t)+ \Av_i(\Rvm,t) \Phi_\Rvm(\rvm,t)] ,
\end{split}
\end{equation}
and then
\begin{equation}
\begin{split}
&
    \nabla_i  \dot{\Phi}_\Rvm(\rvm,t)= -i [\nabla_i \hat{H}^{BO}] \Phi_\Rvm(\rvm,t)  -i \hat{H}^{BO} \nabla_i\Phi_\Rvm(\rvm,t)
 +i \langle \Phi_\Rvm|  \hat{H}^{BO}|\Phi_\Rvm\rangle_\rvm \nabla_i \Phi_\Rvm(\rvm,t) 
   +i \langle \nabla_i \Phi_\Rvm|  \hat{H}^{BO}|\Phi_\Rvm\rangle_\rvm  \Phi_\Rvm(\rvm,t)  \\
& 
    +i \langle \Phi_\Rvm|  \hat{H}^{BO}|\nabla_i \Phi_\Rvm\rangle_\rvm  \Phi_\Rvm(\rvm,t) 
    +i \langle \Phi_\Rvm| [\nabla_i \hat{H}^{BO}]|\Phi_\Rvm\rangle_\rvm  \Phi_\Rvm(\rvm,t) 
    +\langle \Phi_\Rvm |\dot{\Phi}_\Rvm\rangle_\rvm \nabla_i\Phi_\Rvm(\rvm,t) +\Phi_\Rvm(\rvm,t) \nabla_i \langle \Phi_\Rvm |\dot{\Phi}_\Rvm\rangle_\rvm  
 \\
&+  \sum\limits_{j=1}^{N_n} \frac{i}{2 M_j} [ \nabla_i \nabla^2_j\Phi_\Rvm(\rvm,t) -\langle \nabla_i \Phi_\Rvm|\nabla^2_j|\Phi_\Rvm\rangle_\rvm \Phi_\Rvm(\rvm,t)-\langle \Phi_\Rvm|\nabla_i \nabla^2_j|\Phi_\Rvm\rangle_\rvm \Phi_\Rvm(\rvm,t)-\langle \Phi_\Rvm|\nabla^2_j|\Phi_\Rvm\rangle_\rvm \nabla_i \Phi_\Rvm(\rvm,t)]  \\
&+ \! \sum\limits_{j=1}^{N_n} \! \frac{1}{ M_j} \! \nabla_i \! \! \left[\frac{\nabla_j \chi(\Rvm,t)}{\chi(\Rvm,t)} \right] \! \cdot \! [i\nabla_j \Phi_\Rvm(\rvm,t) \! + \! \Av_j(\Rvm,t) \Phi_\Rvm(\rvm,t)] \\
& +  \sum\limits_{j=1}^{N_n} \! \frac{1}{ M_j} \frac{\nabla_j \chi(\Rvm,t)}{\chi(\Rvm,t)} \! \cdot \! [i \nabla_i \nabla_j \Phi_\Rvm(\rvm,t) \! + \! \Phi_\Rvm(\rvm,t) \nabla_i \Av_j(\Rvm,t)  \!+\! \Av_j(\Rvm,t) \nabla_i \Phi_\Rvm(\rvm,t)]  .
\end{split}
\end{equation}
Therefore,
\begin{equation}
\begin{split}
&
    \langle \Phi_\Rvm|\nabla_i \dot{\Phi}_\Rvm\rangle_\rvm=   -i \langle \Phi_\Rvm |\hat{H}^{BO} \nabla_i\Phi_\Rvm\rangle _\rvm
  +i \langle \Phi_\Rvm|  \hat{H}^{BO}|\Phi_\Rvm\rangle_\rvm \langle \Phi_\Rvm|\nabla_i \Phi_\Rvm\rangle_\rvm
   +i \langle \nabla_i \Phi_\Rvm|  \hat{H}^{BO}|\Phi_\Rvm\rangle_\rvm 
    +i \langle \Phi_\Rvm|  \hat{H}^{BO}|\nabla_i \Phi_\Rvm\rangle_\rvm  
 \\
&+\langle \Phi_\Rvm |\dot{\Phi}_\Rvm\rangle_\rvm \langle \Phi_\Rvm|\nabla_i\Phi_\Rvm\rangle_\rvm + \nabla_i \langle \Phi_\Rvm |\dot{\Phi}_\Rvm\rangle_\rvm 
-  \sum\limits_{j=1}^{N_n} \frac{1}{2 M_j} [  i\langle \nabla_i \Phi_\Rvm|\nabla^2_j|\Phi_\Rvm\rangle_\rvm +i\langle \Phi_\Rvm|\nabla^2_j|\Phi_\Rvm\rangle_\rvm \langle \Phi_\Rvm|\nabla_i \Phi_\Rvm\rangle_\rvm] \\
&+ \sum\limits_{j=1}^{N_n} \frac{1}{ M_j} \frac{\nabla_j \chi(\Rvm,t)}{\chi(\Rvm,t)} \cdot [i \langle \Phi_\Rvm|\nabla_i \nabla_j |\Phi_\Rvm\rangle_\rvm \! + \!  \nabla_i \Av_j(\Rvm,t) +\Av_j(\Rvm,t) \langle \Phi_\Rvm|\nabla_i| \Phi_\Rvm\rangle_\rvm] ,
\end{split}
\end{equation}
or
\begin{equation}
\begin{split}
&
    \langle \Phi_\Rvm|\nabla_i \dot{\Phi}_\Rvm\rangle_\rvm= -i \langle \Phi_\Rvm |\hat{H}^{BO} \nabla_i\Phi_\Rvm\rangle_\rvm
  - \langle \Phi_\Rvm|  \hat{H}^{BO}|\Phi_\Rvm\rangle_\rvm \Av_i(\Rvm,t)
   +i \langle \nabla_i \Phi_\Rvm|  \hat{H}^{BO}|\Phi_\Rvm\rangle_\rvm 
    +i \langle \Phi_\Rvm|  \hat{H}^{BO}|\nabla_i \Phi_\Rvm\rangle_\rvm  
 \\
&+i \langle \Phi_\Rvm |\dot{\Phi}_\Rvm\rangle_\rvm \Av_i(\Rvm,t) + \nabla_i \langle \Phi_\Rvm |\dot{\Phi}_\Rvm\rangle_\rvm
-  \sum\limits_{j=1}^{N_n} \frac{1}{2 M_j} [  i\langle \nabla_i \Phi_\Rvm|\nabla^2_j|\Phi_\Rvm\rangle_\rvm -\langle \Phi_\Rvm|\nabla^2_j|\Phi_\Rvm\rangle_\rvm \Av_i(\Rvm,t)] \\
&+ \sum\limits_{j=1}^{N_n} \frac{1}{ M_j} \frac{\nabla_j \chi(\Rvm,t)}{\chi(\Rvm,t)} \cdot [i \langle \Phi_\Rvm|\nabla_i \nabla_j |\Phi_\Rvm\rangle_\rvm+  \nabla_i \Av_j(\Rvm,t) + i \Av_j(\Rvm,t) \Av_i(\Rvm,t)] ,
\end{split}
\end{equation}
or
\begin{equation}
\begin{split}
&
    \langle \Phi_\Rvm|\nabla_i \dot{\Phi}_\Rvm\rangle_\rvm=-i \langle \Phi_\Rvm |\hat{H}^{BO} \nabla_i\Phi_\Rvm\rangle_\rvm
  - \langle \Phi_\Rvm|  \hat{H}^{BO}|\Phi_\Rvm\rangle_\rvm \Av_i(\Rvm,t)
   +i \langle \nabla_i \Phi_\Rvm|  \hat{H}^{BO}|\Phi_\Rvm\rangle_\rvm 
    +i \langle \Phi_\Rvm|  \hat{H}^{BO}|\nabla_i \Phi_\Rvm\rangle_\rvm  
 \\
&+i \langle \Phi_\Rvm |\dot{\Phi}_\Rvm\rangle_\rvm \Av_i(\Rvm,t) + \nabla_i \langle \Phi_\Rvm |\dot{\Phi}_\Rvm\rangle_\rvm
-  \sum\limits_{j=1}^{N_n} \frac{1}{2 M_j} [  i\langle \nabla_i \Phi_\Rvm|\nabla^2_j|\Phi_\Rvm\rangle_\rvm - i \Av_i(\Rvm,t) \nabla_j\Av_j(\Rvm,t) 
+\Av_i(\Rvm,t) \langle \nabla_j \Phi_\Rvm|\nabla_j\Phi_\Rvm\rangle_\rvm ] \\
&+ \sum\limits_{j=1}^{N_n} \frac{1}{ M_j} \frac{\nabla_j \chi(\Rvm,t)}{\chi(\Rvm,t)} \cdot [i \langle \Phi_\Rvm|\nabla_i \nabla_j |\Phi_\Rvm\rangle_\rvm+  \nabla_i \Av_j(\Rvm,t) +i \Av_j(\Rvm,t) \Av_i(\Rvm,t)]  ,
\end{split}
\end{equation}
or
\begin{equation}
\begin{split}
&
    \langle \Phi_\Rvm|\nabla_i \dot{\Phi}_\Rvm\rangle_\rvm=-i \langle \Phi_\Rvm |\hat{H}^{BO} \nabla_i\Phi_\Rvm\rangle_\rvm
  - \langle \Phi_\Rvm|  \hat{H}^{BO}|\Phi_\Rvm\rangle_\rvm \Av_i(\Rvm,t)
   +i \langle \nabla_i \Phi_\Rvm|  \hat{H}^{BO}|\Phi_\Rvm\rangle_\rvm 
    +i \langle \Phi_\Rvm|  \hat{H}^{BO}|\nabla_i \Phi_\Rvm\rangle_\rvm  
 \\
&+i \langle \Phi_\Rvm |\dot{\Phi}_\Rvm\rangle_\rvm \Av_i(\Rvm,t) + \nabla_i \langle \Phi_\Rvm |\dot{\Phi}_\Rvm\rangle_\rvm
-  \sum\limits_{j=1}^{N_n} \frac{1}{2 M_j} [  i\langle \nabla_i \Phi_\Rvm|\nabla^2_j|\Phi_\Rvm\rangle_\rvm - i \Av_i(\Rvm,t) \nabla_j\Av_j(\Rvm,t) 
+\Av_i(\Rvm,t) \langle \nabla_j \Phi_\Rvm|\nabla_j\Phi_\Rvm\rangle_\rvm ] \\
&+ \sum\limits_{j=1}^{N_n} \frac{1}{ M_j} \left[ \frac{1}{2} \frac{\nabla_j |\chi(\Rvm,t)|^2}{ |\chi(\Rvm,t)|^2} +i   \nabla_j S(\Rvm,t) \right] \cdot [i \langle \Phi_\Rvm|\nabla_i \nabla_j |\Phi_\Rvm\rangle_\rvm+  \nabla_i \Av_j(\Rvm,t) +i \Av_j(\Rvm,t) \Av_i(\Rvm,t)]  ,
\end{split}
\end{equation}
or
\begin{equation}
\begin{split}
&
    \langle \Phi_\Rvm|\nabla_i \dot{\Phi}_\Rvm\rangle_\rvm=-i \langle \Phi_\Rvm |\hat{H}^{BO} \nabla_i\Phi_\Rvm\rangle_\rvm
  - \langle \Phi_\Rvm|  \hat{H}^{BO}|\Phi_\Rvm\rangle_\rvm \Av_i(\Rvm,t)
   +i \langle \nabla_i \Phi_\Rvm|  \hat{H}^{BO}|\Phi_\Rvm\rangle_\rvm 
    +i \langle \Phi_\Rvm|  \hat{H}^{BO}|\nabla_i \Phi_\Rvm\rangle_\rvm  
 \\
&+i \langle \Phi_\Rvm |\dot{\Phi}_\Rvm\rangle_\rvm \Av_i(\Rvm,t) + \nabla_i \langle \Phi_\Rvm |\dot{\Phi}_\Rvm\rangle_\rvm
-  \sum\limits_{j=1}^{N_n} \frac{1}{2 M_j} [  i\langle \nabla_i \Phi_\Rvm|\nabla^2_j|\Phi_\Rvm\rangle_\rvm - i \Av_i(\Rvm,t) \nabla_j\Av_j(\Rvm,t) 
+\Av_i(\Rvm,t) \langle \nabla_j \Phi_\Rvm|\nabla_j\Phi_\Rvm\rangle_\rvm ] \\
&+ \sum\limits_{j=1}^{N_n} \frac{1}{ M_j} \left[ \frac{1}{2} \frac{\nabla_j |\chi(\Rvm,t)|^2}{ |\chi(\Rvm,t)|^2} +i   \nabla_jS(\Rvm,t) \right] \cdot [- \nabla_j \Av_i(\Rvm,t)-i \langle \nabla_j \Phi_\Rvm|\nabla_i  \Phi_\Rvm\rangle_\rvm+  \nabla_i \Av_j(\Rvm,t) +i \Av_j(\Rvm,t) \Av_i(\Rvm,t)]  ,
\end{split}
\end{equation}
or
\begin{equation}
\begin{split}
&
    \langle \Phi_\Rvm|\nabla_i \dot{\Phi}_\Rvm\rangle_\rvm=-i \langle \Phi_\Rvm |\hat{H}^{BO} \nabla_i\Phi_\Rvm\rangle_\rvm 
  - \langle \Phi_\Rvm|  \hat{H}^{BO}|\Phi_\Rvm\rangle_\rvm \Av_i(\Rvm,t)
   +i \langle \nabla_i \Phi_\Rvm|  \hat{H}^{BO}|\Phi_\Rvm\rangle_\rvm 
    +i \langle \Phi_\Rvm|  \hat{H}^{BO}|\nabla_i \Phi_\Rvm\rangle_\rvm  
 \\
&+i \langle \Phi_\Rvm |\dot{\Phi}_\Rvm\rangle_\rvm \Av_i(\Rvm,t) + \nabla_i \langle \Phi_\Rvm |\dot{\Phi}_\Rvm\rangle_\rvm
-  \sum\limits_{j=1}^{N_n} \frac{1}{2 M_j} [  i\langle \nabla_i \Phi_\Rvm|\nabla^2_j|\Phi_\Rvm\rangle_\rvm - i \Av_i(\Rvm,t) \nabla_j\Av_j(\Rvm,t) 
+\Av_i(\Rvm,t) \langle \nabla_j \Phi_\Rvm|\nabla_j\Phi_\Rvm\rangle_\rvm ] \\
&+ i \sum\limits_{j=1}^{N_n} \frac{1}{ M_j}   \nabla_j  S(\Rvm,t)  \cdot [- \nabla_j \Av_i(\Rvm,t)-i \langle \nabla_j \Phi_\Rvm|\nabla_i  \Phi_\Rvm\rangle_\rvm+  \nabla_i \Av_j(\Rvm,t) + i \Av_j(\Rvm,t) \Av_i(\Rvm,t)] \\
&+ \sum\limits_{j=1}^{N_n} \frac{1}{ 2 M_j}  \frac{\nabla_j |\chi(\Rvm,t)|^2}{ |\chi(\Rvm,t)|^2} \cdot [- \nabla_j \Av_i(\Rvm,t)-i \langle \nabla_j \Phi_\Rvm|\nabla_i  \Phi_\Rvm\rangle_\rvm+  \nabla_i \Av_j(\Rvm,t) +i \Av_j(\Rvm,t) \Av_i(\Rvm,t)] .
\end{split}
\end{equation}
Then
\begin{equation}
\begin{split}
&
  \Im  \langle \Phi_\Rvm|\nabla_i \dot{\Phi}_\Rvm\rangle_\rvm=-\Re \langle \Phi_\Rvm |\hat{H}^{BO} \nabla_i\Phi_\Rvm\rangle_\rvm 
   +2 \Re \langle \nabla_i \Phi_\Rvm|  \hat{H}^{BO}|\Phi_\Rvm\rangle_\rvm 
 + \Im \nabla_i \langle \Phi_\Rvm |\dot{\Phi}_\Rvm\rangle_\rvm\\
&-  \Re \sum\limits_{j=1}^{N_n} \frac{1}{2 M_j} [  \langle \nabla_i \Phi_\Rvm|\nabla^2_j|\Phi_\Rvm\rangle_\rvm -  \Av_i(\Rvm,t) \nabla_j\Av_j(\Rvm,t) 
] \\
&+ \Re \sum\limits_{j=1}^{N_n} \frac{1}{ M_j}   \nabla_j  S(\Rvm,t)  \cdot [- \nabla_j \Av_i(\Rvm,t)-i \langle \nabla_j \Phi_\Rvm|\nabla_i  \Phi_\Rvm\rangle_\rvm+  \nabla_i \Av_j(\Rvm,t) ] \\
&+\Re  \sum\limits_{j=1}^{N_n} \frac{1}{ 2 M_j}  \frac{\nabla_j |\chi(\Rvm,t)|^2}{ |\chi(\Rvm,t)|^2} \cdot [- \langle \nabla_j \Phi_\Rvm|\nabla_i  \Phi_\Rvm\rangle_\rvm+ \Av_j(\Rvm,t) \Av_i(\Rvm,t)] .
\end{split}
\end{equation}
We integrate
\begin{equation}
\begin{split}
&
  \int d\Rvm |\chi(\Rvm,t)|^2  \Im  \langle \Phi_\Rvm|\nabla_i \dot{\Phi}_\Rvm\rangle_\rvm=  \int d\Rvm |\chi(\Rvm,t)|^2 \left\{-\Re \langle \Phi_\Rvm |\hat{H}^{BO} \nabla_i\Phi_\Rvm\rangle_\rvm 
   +2 \Re \langle \nabla_i \Phi_\Rvm|  \hat{H}^{BO}|\Phi_\Rvm\rangle_\rvm 
 + \Im \nabla_i \langle \Phi_\Rvm |\dot{\Phi}_\Rvm\rangle_\rvm \right. \\
&-  \Re \sum\limits_{j=1}^{N_n} \frac{1}{2 M_j} [  \langle \nabla_i \Phi_\Rvm|\nabla^2_j|\Phi_\Rvm\rangle_\rvm -  \Av_i(\Rvm,t) \nabla_j\Av_j(\Rvm,t) 
] \\
&\left. + \Re \sum\limits_{j=1}^{N_n} \frac{1}{ M_j}   \nabla_j  S(\Rvm,t)  \cdot [- \nabla_j \Av_i(\Rvm,t)-i \langle \nabla_j \Phi_\Rvm|\nabla_i  \Phi_\Rvm\rangle_\rvm+  \nabla_i \Av_j(\Rvm,t) ]\right\} \\
&+\Re  \int d \Rvm \sum\limits_{j=1}^{N_n} \frac{1}{ 2 M_j}  \nabla_j |\chi(\Rvm,t)|^2 \cdot [- \langle \nabla_j \Phi_\Rvm|\nabla_i  \Phi_\Rvm\rangle_\rvm+ \Av_j(\Rvm,t) \Av_i(\Rvm,t)] ,
\end{split}
\end{equation}
and, after simplification,
\begin{equation}
\begin{split}
&
  \Im  \int d\Rvm |\chi(\Rvm,t)|^2 \langle \Phi_\Rvm|\nabla_i \dot{\Phi}_\Rvm\rangle_\rvm= \int d\Rvm |\chi(\Rvm,t)|^2 \left\{
    \Re \langle \nabla_i \Phi_\Rvm|  \hat{H}^{BO}|\Phi_\Rvm\rangle_\rvm 
 + \Im \nabla_i \langle \Phi_\Rvm |\dot{\Phi}_\Rvm\rangle_\rvm \right. 
+  \Re \sum\limits_{j=1}^{N_n} \frac{1}{2 M_j}  
\langle \nabla_i \nabla_j \Phi_\Rvm|\nabla_j\Phi_\Rvm\rangle_\rvm 
 \\
&\left. + \Re \sum\limits_{j=1}^{N_n} \frac{1}{ M_j}   \nabla_j  S(\Rvm,t)  \cdot [- \nabla_j \Av_i(\Rvm,t)-i \langle \nabla_j \Phi_\Rvm|\nabla_i  \Phi_\Rvm\rangle_\rvm+  \nabla_i \Av_j(\Rvm,t) -\Re \sum\limits_{j=1}^{N_n} \frac{1}{ 2 M_j}     
\Av_j(\Rvm,t) \nabla_j \Av_i(\Rvm,t) ]\right\} .
\end{split}
\end{equation}
Then
\begin{equation}
\begin{split}
&\int \! d\Rvm |\chi(\Rvm,t)|^2  \! \left[ \! -\nabla_i \epsilon(\Rvm,t) \! + \! \frac{\pa \Av_i(\Rvm,t)}{\pa t}\right] \! = \!
\int \! d\Rvm |\chi(\Rvm,t)|^2\left\{- \langle \Phi_\Rvm|\hat{H}^{BO}|\nabla_i  \Phi_\Rvm\rangle_\rvm \! - \! \langle\nabla_i \Phi_\Rvm|\hat{H}^{BO}| \Phi_\Rvm\rangle_\rvm \! - \! \langle \Phi_\Rvm|[\nabla_i\hat{H}^{BO}]|  \Phi_\Rvm\rangle_\rvm  \right. \\
& \left. - \sum\limits_{j=1}^{N_n} \left[ \frac{1}{ M_j} \Re \langle \nabla_i\nabla_j\Phi_\Rvm|\nabla_j \Phi_\Rvm\rangle_\rvm 
-\frac{\nabla_i \Av_j(\Rvm,t)\cdot \Av_j(\Rvm,t)}{ M_j} \right] 
 +2 \Im  \langle \Phi_\Rvm|  \nabla_i \dot{\Phi}_\Rvm\rangle_\rvm -2 \Im \nabla_i \langle \Phi_\Rvm| \dot{\Phi}_\Rvm\rangle_\rvm
 \right\}\\
&= \int d\Rvm |\chi(\Rvm,t)|^2\left\{- \langle \Phi_\Rvm|\hat{H}^{BO}|\nabla_i  \Phi_\Rvm\rangle_\rvm - \langle\nabla_i \Phi_\Rvm|\hat{H}^{BO}| \Phi_\Rvm\rangle_\rvm - \langle \Phi_\Rvm|[\nabla_i\hat{H}^{BO}]|  \Phi_\Rvm\rangle_\rvm \right. \\
& \left. - \sum\limits_{j=1}^{N_n} \left[ \frac{1}{ M_j} \Re \langle \nabla_i\nabla_j\Phi_\Rvm|\nabla_j \Phi_\Rvm\rangle_\rvm 
-\frac{\nabla_i \Av_j(\Rvm,t)\cdot \Av_j(\Rvm,t)}{ M_j} \right] 
  -2 \Im \nabla_i \langle \Phi_\Rvm| \dot{\Phi}_\Rvm\rangle_\rvm
 \right\}\\
&+2 \int d\Rvm |\chi(\Rvm,t)|^2 \left\{
    \Re \langle \nabla_i \Phi_\Rvm|  \hat{H}^{BO}|\Phi_\Rvm\rangle_\rvm 
 + \Im \nabla_i \langle \Phi_\Rvm |\dot{\Phi}_\Rvm\rangle_\rvm \right. 
+  \Re \sum\limits_{j=1}^{N_n} \frac{1}{2 M_j}  
\langle \nabla_i \nabla_j \Phi_\Rvm|\nabla_j\Phi_\Rvm\rangle_\rvm 
 \\
&\left. + \Re \sum\limits_{j=1}^{N_n} \frac{1}{ M_j}   \nabla_j  S(\Rvm,t)  \cdot [- \nabla_j \Av_i(\Rvm,t)-i \langle \nabla_j \Phi_\Rvm|\nabla_i  \Phi_\Rvm\rangle_\rvm+  \nabla_i \Av_j(\Rvm,t) -\Re \sum\limits_{j=1}^{N_n} \frac{1}{ 2 M_j}     
\Av_j(\Rvm,t) \nabla_j \Av_i(\Rvm,t) ]\right\},
\end{split}
\end{equation}
which, after extensive simplifications, can be written as
\begin{equation}
\begin{split}
&\int d\Rvm |\chi(\Rvm,t)|^2  \left[-\nabla_i \epsilon(\Rvm,t)+\frac{\pa \Av_i(\Rvm,t)}{\pa t}\right]  
= \int d\Rvm |\chi(\Rvm,t)|^2\left\{ - \langle \Phi_\Rvm|[\nabla_i\hat{H}^{BO}]|  \Phi_\Rvm\rangle_\rvm \right. \\
& \left. + \sum\limits_{j=1}^{N_n} 
\frac{\Av_j(\Rvm,t) [\nabla_i \Av_j(\Rvm,t) -\nabla_j \Av_i(\Rvm,t)]}{ M_j} 
 + \ \sum\limits_{j=1}^{N_n} \frac{2}{ M_j}   \nabla_j  S(\Rvm,t)  \cdot [\nabla_i \Av_j(\Rvm,t)- \nabla_j \Av_i(\Rvm,t)+\Im \langle \nabla_j \Phi_\Rvm|\nabla_i  \Phi_\Rvm\rangle_\rvm ]\right\}.
\end{split}
\end{equation}

\end{widetext}

Taking the classical limit in the last equation and  using  Lemmas \ref{LL} and \ref{LLL}, we arrive at Eq.~\eqref{Llast}.

\end{document}